\documentclass[journal,draft,onecolumn,12pt]{IEEEtran}
\usepackage{amsfonts,color,morefloats}
\usepackage{amssymb,amsmath,latexsym,amsthm}
\usepackage{float}
\usepackage{algorithm}
\usepackage{algorithmic}
\usepackage{multirow,tabu}
\usepackage{nicematrix} 
\usepackage{amsmath}    
\usepackage{multirow, bigdelim} 
\usepackage{enumitem}   

\newcommand{\C}{{\mathcal{C}}}
\newcommand{\D}{{\mathcal{D}}}

\newtheorem{theorem}{Theorem}
\newtheorem{lemma}[theorem]{Lemma}

\newtheorem{corollary}[theorem]{Corollary}

\newtheorem{example}[theorem]{Example}
\newtheorem{remark}[theorem]{Remark}
\definecolor{myred}{RGB}{255,0,0}

\begin{document}

\title{{ Explicit construction of} optimal binary linear complementary pairs via Belov codes \thanks{
  This work was supported by the National Natural Science Foundation of China (No. 62372247) and by open research fund of State Key Laboratory of Cyberspace Security Defense (No. 2025-MS-03).
}
}

\author{Yansheng Wu,\thanks{Y. Wu is with the School of Computer Science, Nanjing University of Posts and Telecommunications, Nanjing, 210023, China, and the State Key Laboratory of Cyberspace Security Defense (Institute of Information Engineering, Chinese Academy of Sciences, Beijing 100085), China. 
 Email: yanshengwu@njupt.edu.cn.} 
\and Jing Chen,\thanks{J. Chen is with the School of Computer Science, Nanjing University of Posts and Telecommunications, Nanjing
210023, China. Email: 1806108644@qq.com.} 
\and  Jong Yoon Hyun\thanks{
J. Y. Hyun is with the Konkuk University, Glocal Campus, 268 Chungwon-daero Chungju-si Chungcheongbuk-do 27478, South Korea, Email: hyun33@kku.ac.kr.}
}

\date{\today}
\maketitle

\begin{abstract} 

The study of linear complementary pairs (LCPs) has been a hot topic in coding theory. Recently, Güneri constructed an infinite family of optimal binary LCPs from Solomon-Stiffler codes. Specifically, he proved that for every $k \geq 5$ and $d \geq \lceil (k - 1)/2 \rceil 2^{k - 1}$, there exist optimal binary linear complementary pairs of codes of length $g_2(k,d)$, where $g_2(k,d)$ denotes the Griesmer bound.
 In this paper, we construct additional optimal binary LCPs from Belov codes, which form a generalized class of Solomon-Stiffler codes. Our results not only generalize Güneri's construction, but also enable  more flexible construction of  optimal binary LCPs with  a wide range of minimum distances.


\end{abstract}


\begin{IEEEkeywords}
Linear complementary pairs, Belov codes, Solomon-Stiffler codes, Griesmer bound
\end{IEEEkeywords}

\section{Introduction}\label{sec-intro} 

A pair of linear codes $(\C,\D)$ of length $n$ over $\mathbb{F}_{q}$ is called a linear complementary pair (LCP) if their direct sum $\C \oplus \D$ yields the full space $\mathbb{F}_{q}^{n}$. Since their dimensions must satisfy $\dim \C + \dim \D = n$ and $\C \cap \D = \{\mathbf{0}\}$, if $\dim \C = k$ and $\dim \D = n - k$, then we call $(\C,\D)$ an $[n,k]$ LCP. It is easy to show that the concept of LCPs is a generalization of linear complementary dual codes (LCD codes), see \cite{c18}. The most important problem in this research  is the construction of LCPs with the maximum security parameter. For given $n$ and $k$, let
\[
d_{LCP}(n,k) = \max \left\{ \min \left\{ d(\C), d(\D^{\bot}) \right\} : (\C,\D) \text{ is an } [n,k] \text{ binary LCP} \right\}
\]
denote the maximum security parameter for an $[n,k]$ binary LCP.  If $d_L(n,k)$ denotes the highest minimum distance among binary linear $[n,k]$ codes, we clearly have
$
 d_{LCP}(n,k) \leq d_L(n,k).
$
We call an $[n,k]$ linear code optimal if its minimum distance is $d_L(n,k)$. We also call an $[n,k]$ LCP $(\C,\D)$ optimal if the {maximum} security parameter  reaches $d_L(n,k)$.  

 LCPs have been used in the framework of direct sum masking, which was proposed as a countermeasure against side channel attacks (SCAs) and fault injection attacks (FIAs), see \cite{G23,Carlet18, G19, H21, Lixia, L22, L231, L232}. Here we list some recent progress on LCPs.

\begin{itemize}
\item  Carlet {\it et al.} \cite{Carlet19} proved that for any $[n,k,d]$ linear code $\C$ over $\mathbb{F}_{q}$, with $q > 2$, there is a monomial transformation $\sigma$ on $\mathbb{F}_{q}^{n}$ such that $(\C,\sigma(\C)^{\bot})$ is an LCP. Consequently, the security parameter in this configuration is $d$, since $d = d(\C) =d(\sigma(\C)^{{\perp}})$. { This fully settles the problem of determining the optimal security parameter for $q$-ary LCPs when \(q \geq 2\).}  {In the binary case, however, the authors demonstrated that given an $[n,k,d]$ binary code, one can construct an $[n,k]$ LCP of binary codes $(\C,\D)$ satisfying $d(\C) = d(\D^{\bot}) \geq d - 1$. Based on this construction, they established the following bounds: $d_L(n,k) - 1\leq d_{LCP}(n,k) \leq d_L(n,k)$. As a result, the exact determination of the optimal security parameter for binary LCPs remains an open problem and is not yet explicitly settled.}

\item

{  Choi et al.~\cite{Choi23} initiated the systematic investigation of 
optimal binary LCPs of codes. They provided explicit algebraic constructions of such 
codes to prove that for any $k \geq 2$ and $n \equiv 0 \text{ or } 1 \pmod{2^k - 1}$, 
the maximum security parameter satisfies $d_{\text{LCP}}(n, k) = d_L(n, k)$~\cite[Theorem~15]{Choi23}. 
Furthermore, except for $(n, k) = (4, 3)$ and $(n, k) = (8, 4)$, they explicitly 
constructed optimal binary LCPs achieving $d_{\text{LCP}}(n, k) = d_L(n, k)$ for 
all $2 \leq k \leq 4$ and all $n$~\cite[Remark~1]{Choi23}.}

\item {For $k \geq 5$, Güneri~\cite{G23} explicitly constructed optimal binary LCPs from Solomon-Stiffler (SS) codes; specically, the author proved that for every $k \geq 5$ and minimum distance $d \geq \lceil (k - 1)/2 \rceil 2^{k - 1}$, there exists an optimal binary LCP of length $g_2(k,d)$, where $g_2(k,d)$ denotes the Griesmer bound.}

\item Li \textit{et al.} characterized the optimal value \(d_{\mathrm{LCP}}(n,k)\) in terms of the intersection dimension \(\dim(\C_1\cap \C_2^\perp)\) and the existence of a permutation automorphism \(\sigma\in \mathrm{PAut}(\mathbb{F}_2^n)\) such that \((\C_1,(\sigma(\C_2))^\perp)\) is an \(\ell\)-intersection pair. This characterization resolved the open problem of Carlet \textit{et al.}~\cite{Carlet19} and confirmed the conjecture of Choi \textit{et al.}~\cite{Choi23}. However, their proof is existential rather than constructive, and  their results settle the existence and characterization questions, but the corresponding constructive problem remains open.

\end{itemize}


{ It is well known that Belov codes form a generalized class of Solomon-Stiffler (SS) codes (\cite{S65, T92}). In this paper, based on Güneri's work, we obtain a broader class of optimal binary LCPs from Belov codes.

The remainder of this paper is organized as follows: 
     Section II recalls the Griesmer bound,  Solomon-Stiffler codes, and Belov codes. Section III reviews the main idea of Güneri.  Section IV obtains more LCPs by using Belov codes and the Güneri’s construction. In Section V, we present 
     our main results on new construction of binary optimal LCPs by using SS codes and Belov codes. 
     Section VI concludes this paper.

}


\section{Preliminaries}

In this section, we will recall some need results. 

\subsection{Griesmer bound}
Let $q$ be a prime power  and $\mathbb{F}_{q}$ be the finite field with $q$ elements. Let $\mathbb{F}_{q}^{n}$ be an $n$-dimensional vector space over $\mathbb{F}_{q}$. Let $\mathcal{C}$ be a linear code over $\mathbb{F}_q$ with parameters $[n,k,d]_q$, where $n$ is the length of the code, $k$ is its dimension and $d$ is the minimum distance. In coding theory, constructing optimal codes is essential for improving the reliability of communication systems. A linear code $\mathcal{C}$ with parameters $[n,k,d]_q$ is called optimal if its minimum distance $d$ meets a known theoretical bound. See \cite{HP03} for the detail. {Two binary linear codes are equivalent if one code is obtained from the other by permuting the coordinates.} 

\begin{lemma}{ \cite{G}  \label{Griesmer}
(Griesmer bound)  Let $\mathcal{C}$ be an $[n,k,d]_{q}$-code. Then
$$n \geq g_q(k,d)= \sum\limits_{i=0}^{k-1}\bigg\lceil\frac{d}{q^i}\bigg\rceil,$$  
where $\lceil{\cdot}\rceil$ denotes the ceiling function.  }
\end{lemma}

A code is called a {\it Griesmer code} if it attains the Griesmer bound with equality. By using the Griesmer bound, it can be verified that every Griesmer code  
is  dimension-optimal and distance-optimal simultaneously.

\subsection{Solomon-Stiffler codes}
In this subsection, we will recall binary Solomon-Stiffler (SS) codes.
 
Let ${G} = [S_{k}]$ be a $k \times (2^{k} - 1)$ matrix whose columns comprise all nonzero vectors in $\mathbb{F}_{2}^{k}$. Then ${G}$ generates the $[2^{k} - 1, k, 2^{k - 1}]$ simplex code where all the codewords have weight $2^{k - 1}$. We denote the matrix
\[
 \overbrace{  \left[ S_k | \cdots | S_k  \right] }^{h ~\text{times}}  
\]
by $S_{k}(h)$, {where $S_k(1)=S_k$.} The code generated by this matrix, which we will denote by $\mathcal C_{{S}_{k}(h)}$, is the $h$-times replicated simplex code. It is clear that $\mathcal C_{{S}_{k}(h)}$ is also a constant-weight code with parameters $[h(2^{k} - 1), k, h2^{k - 1}]$.

Let
\[
F=\bigcup_{i=1}^{t}U_i,
\qquad
1\le u_1<u_2<\cdots<u_t<k.
\]
For each \(i=1,\dots,t\), let \(U_i\) be a subspace of \(\mathbb F_2^k\) with
$
\dim_{\mathbb F_2}U_i=u_i.
$
By saying that \(U_1,\dots,U_t\) are \emph{disjoint subspaces}, we mean that they intersect trivially as follows:
\[
U_i\cap U_j=\{\mathbf{0}\}\qquad (i\ne j).
\]

For each \(i\), list the \(2^{u_i}-1\) nonzero vectors of \(U_i\) in some fixed order:
\[
U_i\setminus\{\mathbf{0}\}
=
\{v_{i,1},v_{i,2},\dots,v_{i,2^{u_i}-1}\}.
\]
Now form the binary matrix {
\[
M_i=
\begin{bmatrix}
v^T_{i,1} & v^T_{i,2} & \cdots & v^T_{i,2^{u_i}-1},
\end{bmatrix},
\]
where $T$ denotes the matrix transpose.} Hence \(M_i\) is a
$
k\times (2^{u_i}-1)
$
matrix over \(\mathbb F_2\). Since its columns span the subspace \(U_i\), its rank is
$
\operatorname{rank}(M_i)=\dim U_i=u_i.
$
{In the original notation, the same symbol \(U_i\) is used both for the subspace and for this matrix; it is clearer to denote the matrix by \(M_i\).

Finally, every nonzero row of \(M_i\) has weight \(2^{u_i-1}\). Indeed, a row of \(M_i\) is obtained by evaluating on the nonzero vectors of \(U_i\) some coordinate functional
$
\varphi:\mathbb F_2^k\to \mathbb F_2.
$
If \(\varphi|_{U_i}=0\), then the row is zero. Otherwise, \(\varphi|_{U_i}\) is a nonzero linear functional on \(U_i\), so its kernel is a hyperplane in \(U_i\), of size \(2^{u_i-1}\). Hence exactly \(2^{u_i-1}\) vectors of \(U_i\) are mapped to \(0\), and exactly \(2^{u_i-1}\) are mapped to \(1\). Since the zero vector is not a column of \(M_i\), each nonzero row contains exactly \(2^{u_i-1}\) ones. Therefore its weight is
$
2^{u_i-1}.
$
}
In compact form,
$
U_i\le \mathbb F_2^k,~ \dim U_i=u_i,~ U_i\cap U_j=\{\mathbf{0}\} \textrm{ for all } i\ne j,
$
and
\[
M_i=[\,U_i\setminus\{\mathbf{0}\}\,]\in \mathbb F_2^{\,k\times (2^{u_i}-1)},~
\operatorname{rank}(M_i)=u_i,
\]
where the columns of \(M_i\) are precisely the nonzero vectors of \(U_i\).

 Clearly, $U_i$ is contained in $S_k$ for all $i$, and the nonzero rows of such a matrix have weight $2^{u_i - 1}$. 
Provided that $U_1, \ldots, U_t$ are chosen to satisfy
\begin{equation}\label{eq1}
\sum_{i = 1}^{t} u_i \leq hk 
\end{equation}
distinct basis elements can be selected for all $U_i$ from the columns of $S_k(h)$, which contains $hk$ basis elements in its columns. Consequently, $F$ can be formed by the columns derived from the nonzero elements of the $U_i$, as condition in \eqref{eq1} ensures that no column in $S_k(h)$ is punctured more than $h$ times.


\begin{lemma}\label{lem2} (\cite{B74, S65})
{ If $1 \leq u_1 < u_2 < \cdots < u_t < k$
 with $\sum_{i=1}^{t}u_i\le hk$}, then ${G} = [S_{k}(h) \setminus F]$ generates  a binary  SS code with parameters
\begin{equation}\label{eqss}
    [n,k,d]=[ h(2^k - 1) - \sum_{i = 1}^{t} (2^{u_i} - 1),k,
 h2^{k - 1} - \sum_{i = 1}^{t} 2^{u_i - 1}].
\end{equation}
This code meets the Griesmer bound.
\end{lemma}

\subsection{Belov codes}
In this subsection, we will recall binary Belov codes, which is a generalization of SS codes. SS codes and Belov codes are two well-known Griesmer codes, see \cite{T81}-\cite{T03}. 

Let $F_1 = W \setminus A,$ where $W$ is a subspace of $\mathbb{F}_2^k$ with $\dim W = \theta \geq 3$, $A \subset W$, $\theta - 1 < |A| \leq \theta + 1$ and any set of $\leq \theta$ points in $A$ are linearly independent. Belov codes additionally remove $F_1 = W \setminus A$, which yield a larger family of parameters that covers cases unreachable by SS codes.

\begin{lemma}\label{lem3}(\cite{B74, T92})
If { $3\le \theta\le u_1<\cdots< u_t<k$
 with $\sum_{i=1}^{{min}\{h+1,t\}}u_{t+1-i}\le hk$} and $F \cap F_1 = \emptyset$, then ${G} = [S_{k}(h) \setminus (F \cup F_1)]$ generates a binary Belov code with parameters
\begin{equation}
    [n,k,d]=[h(2^k - 1) - \sum_{i = 1}^{t} (2^{u_i} - 1) - \sum_{i = 1}^{\theta - 2} ( 2^{i + 1} - 1 )- \varepsilon,k,h2^{k - 1} - \sum_{i = 1}^{t} 2^{u_i - 1} - \sum_{i = 1}^{\theta - 2} 2^{i} - \varepsilon]
\end{equation}
where $0 \leq \varepsilon = \theta + 1 - |A| \leq 1$ and $F$ was defined in Section II. B.
This code meets the Griesmer bound.
\end{lemma}

We present a numerical example to show that the Belov construction possesses more flexible parameters than the Solomon--Stiffler (SS) construction.

\begin{example}
{\rm
Let \(h=1\) and \(k=6\). Then the corresponding binary SS codes have parameters
$
[48,6,24] \text{ for } t=1,\ u_1=4,
$
and
$
[45,6,22] \text{ for } t=2,\ u_1=2,\ u_2=4,
$
respectively. Belov~\cite{B742} gave a new construction of Griesmer codes by deleting more columns based on SS codes. Using Magma, for \(h=1\), \(k=6\), \(t=1\), and \(u_1=4\), we obtain the following table.
\[
\begin{array}{c|cc}
\theta & \varepsilon=0 & \varepsilon=1\\ \hline
3 & [45,6,22] & [44,6,21]\\
4 & [38,6,18] & [37,6,17]
\end{array}
\]
Thus the corresponding binary Belov codes have parameters \([45,6,22]\) and \([44,6,21]\) when \(\theta=3\), and \([38,6,18]\) and \([37,6,17]\) when \(\theta=4\), for \(\varepsilon=0\) and \(\varepsilon=1\), respectively.
}
\end{example}

  We can construct SS codes that share the same parameters as those from the Belov construction but are not equivalent to them.

\begin{example}
{\rm
Let \(h=2\) and \(k=5\) . The Solomon-Stiffler construction {used by  Güneri} with
\[
t=4, (u_1,u_2,u_3,u_4)=(1,2,3,4)
\]
yields a binary \([36,5,17]\) code. On the other hand, the Belov construction starting from the Solomon--Stiffler code with
$
t=2, (u_1,u_2)=(3,4)
$
and then deleting the additional set \(F_1=W\setminus A\) with \(\theta=3\) and \(|A|=3\) also yields a binary \([36,5,17]\) code. {Their parameters are identical, but the two binary codes are not  equivalent:} the former deletes the union of subspaces of dimensions \(1,2,3,4\), whereas the latter deletes subspaces of dimensions \(3,4\) together with the set \(W\setminus A\). The deleted column sets have different structures, and a direct computation shows the two codes have different weight distributions. Hence they are inequivalent.   
}
\end{example}





 Belov \cite{B742} presented the  necessary and sufficient condition for the existence of SS codes and Belov codes.

\begin{lemma}(\cite[Theorem 3]{B742}, \cite{T83})\label{lem4} For any two integers $ k$ and $d = h2^{k - 1} - \sum_{i = 1}^{t} 2^{u_i - 1}$ with $1 \leq u_1 < u_2 < \cdots < u_t < k$, there exists a $[g_2(k,d), k,d]$ of the Solomon-Stiffler type or of the Belov type if and only if either 
 $$\sum_{i=1}^{{min}\{h+1,t\}}u_{t+1-i}\le hk \mbox{ or } u_{i+1} - u_{i} = 1, ~ \mbox{ for } i=1, 2, \ldots,t-1, ~\mbox{ and }u_1 \in \{1,2\}.$$
\end{lemma}

 \section{LCPs constructed from SS codes: Güneri's idea}


In this section, we will review the key idea behind the construction of optimal binary LCPs  in {\cite{Choi23,G23,Li25}. 

Firstly, in order to review the construction method of \cite[Theorems 16, 17, 18]{Choi23},}
let \(\C\) and \(\D\) be binary linear codes generated by $G = [I_k : M]$ and $H = [N : I_{n-k}]$, respectively. Then \((\C, \D)\) forms an LCPs if and only if \(I_k + MN\) is invertible \cite[Proposition 13]{Choi23}. Note that the dual code \(\D^\perp\) is generated by \(\widetilde{H} = [I_k : N^T]\). 
The construction begins with a generator matrix \(G\) of an optimal code. If \(N^T\) can be obtained from \(M\) by permuting its rows and columns, then \(\D^\perp\) shares the same parameters as \(\C\), and hence is also optimal. The goal is then to find suitable permutations yielding such an \(N^T\) for which \(I_k + MN\) is invertible. 
Using this approach, it has been shown that families of optimal binary LCPs  can be constructed in two cases: when \(n \equiv 0, 1 \pmod{2^k - 1}\) for any \(k \geq 2\) \cite[Theorem 15]{Choi23}, and for \(2 \leq k \leq 4\) and all \(n\), with exactly two exceptions \cite[Theorems 16, 17, 18]{Choi23}.


Next, we describe the called maximal anticode, which was introduced in \cite{G23} for constructing optimal LCP from Solomon‑Stiffler codes. Given $k\ge 5$, a maximal anticode corresponds to an anticode $F$ assembled by collecting all nonzero column vectors from a nested sequence of disjoint subspaces $U_1,U_2,\ldots,U_{k-3}\subset \mathbb{F}_2^k$
with strictly increasing dimensions $u_1=3,\ u_2=4,\ldots, u_{k-3}=k-1$.
Each subspace $U_i$ contributes its full set of nonzero columns to the anticode $F$.
When the condition $\sum_{i=1}^{k-3} u_i \le hk$ holds, these subspaces can be chosen such that no column of $S_k(h)$ is punctured more than $h$ times.
For this maximal anticode, the resulting Solomon‑Stiffler generator matrix $G = S_k(h)\setminus F$ is guaranteed to contain two copies of the $k\times k$ identity matrix $I_k$ as well as the special columns $x_1=(1,1,\ldots,1)^T,x_2=(0,1,\ldots,1)^T$ defined in Lemma 3 of \cite{G23}.

To construct binary optimal LCPs, Güneri \cite{G23} chose $\C$ to be a binary SS  code, obtained by deleting  a set $F$ of columns from the generator matrix $S_k(h)$, resulting in $G=[S_{k}(h) \setminus F]$. Main steps are listed as follow:

(1) Consider the case of maximal anticode for $u_1 \geq 3$ in the sense that \({G}^{\prime }\) consists of nonzero columns generated by subspaces $U_1,U_2,\ldots ,U_{k-3} \subset  {\mathbb{F}}_{2}^{k}$ with dimensions $u_1 = 3,u_2 = 4,\ldots ,u_{k - 4} = k-2,u_{k - 3} = k-1$.  First define the subspace $U_{k-3}$ of largest dimension:
\begin{equation}\label{eq2}
{U}_{k-3} \mathrel{\text{ := }} \left\{  {\left( {{a}_{1},\ldots ,{a}_{k}}\right)  \in  {\mathbb{F}}_{2}^{k} : \mathop{\sum }\limits_{{i = 1}}^{k}{a}_{i} = 0}\right\}.
\end{equation}
The following is a basis for $U_{k-3}$:  
\begin{equation}\mathcal{B} = \{ \left( {1,0,\ldots ,0,1}\right) ,\left( {0,1,0,\ldots ,0,1}\right) ,\ldots ,\left( {0,\ldots ,0,1,1}\right) \}.\end{equation}
Then  the second largest subspace  $U_{k-4}$ has dimension $k-2$, which it is spanned by
$\mathcal{B} \setminus  \{ \left( {0,1,0,\ldots ,0,1}\right) \}$. To make sure that there are enough copies $S_k$ to delete those subspaces, the condition $d \geq \lceil (k - 1)/2 \rceil 2^{k - 1}$ is required. Hence $\sum_{i=1}^{t}u_i\le hk$ holds and the induced  SS code is a Griesmer code.

(2) The desired columns (i.e. two copies of $I_k$, $x_1=(1,1,\ldots,1)^T$ and $x_2=(0,1,\ldots,1)^T$ ) are guaranteed to exist in $G$ at the end of forming the maximal anticode without further puncturing in the first two copies of $S_k$. This matrix $G$ is then transformed into systematic form $[I_k : M]$, where $M$ is a $k\times (n-k)$ matrix.  In the SS construction,  $MM^T = I_k$ when $u_1 \geq 3$.

(3) The final step is to find, from $M$, an  $(n-k)\times k$ matrix $N$ such that $I_k+MN$ is invertible. This guarantees that the code \(\D\) generated by
$
[N:I_{n-k}]
$
is complementary to \(\mathcal C\), i.e.,
\[
\mathcal C\cap \D=\{0\},\qquad \mathcal C+\D=\mathbb F_2^n.
\]
Moreover, \(N\) is chosen so that
$
d(\D^\perp)=d(\mathcal C).
$
It should be mentioned that the optimality is achieved precisely in this step. Indeed, since \(\mathcal C\) is a Griesmer SS code, we have
$
d(\mathcal C)=d_L(n,k).
$
Therefore
$
d_{\mathrm{LCP}}(\mathcal C,\D)
=
d_L(n,k)
=
d_{\mathrm{LCP}}(n,k),
$
so the resulting binary LCP is optimal.

Based on this {explicit construction} approach, Güneri \cite{G23} proved the following theorem by using that for any $k$ and $d \geq \lceil (k - 1)/2 \rceil 2^{k - 1}$, there is a binary $[g_2(k,d),k,d]$ SS code, see \cite{BM}.

\begin{theorem}\cite{G23}\label{sun-lem14}
For any $k \geq 5$ and $d \geq \lceil (k - 1)/2 \rceil 2^{k - 1}$, there exists an optimal binary $[g_2(k,d),k]$ LCP  $(\C,\D)$ with the security parameter $d$.
\end{theorem}

{Finally,} Li, Shi, and Ling \cite{Li25} proved that if there exists a binary linear $[n,k,d]$ code $\C$ with the all-one vector ${\bf 1}\notin \mbox{Hull}(\C):=\C \cap \C^{\bot}$, then there exists a binary $[n,k]$ LCP  with security parameters $d$. 

\begin{lemma} \cite[Theorem 16]{Li25} \label{lem6}
    Assume that $k\ge 2$, $d = h2^{k - 1} - \sum_{i = 1}^{t} 2^{u_i - 1}$ with $1 \leq u_1 < u_2 < \cdots < u_t < k$, and $\sum_{i=1}^{{min}\{h+1,t\}}u_{t+1-i}\le hk $.
      Let $\C$ be a SS code with the parameters $[g_2(k,d), k, d]$ in Lemma \ref{lem4}. Then ${\bf 1} \in \mbox{Hull}(\C)$ if and only if $h=1, t=1, $ and $u_1=k-1$.

\end{lemma}


The following theorem is obtained by combining Lemmas \ref{lem4} and \ref{lem6}, 

\begin{theorem}{\cite[Corollary 15, Theorem 16]{Li25}}\label{thm6}
     Assume that $k\ge 2$, $d = h2^{k - 1} - \sum_{i = 1}^{t} 2^{u_i - 1}$ with $1 \leq u_1 < u_2 < \cdots < u_t < k$, and $\sum_{i=1}^{{min}\{h+1,t\}}u_{t+1-i}\le hk $.
     If $(h,t,u_1)\neq (1,1,k-1)$, then there exists a binary $[g_2(k,d), k]$ LCP  with security parameters $d$. { As a direct consequence of this result, we obtain Theorem \ref{sun-lem14}.}
\end{theorem}

{ Observe that Theorem \ref{thm6} gives an existence result because of Lemma \ref{lem6}. The problem of finding a specific construction method to generate LCPs from a given linear code remains difficult. In the next section, we will combine Güneri's idea  with Belov codes to construct a broad new class of optimal binary LCPs.}

{
\section{Optimal LCPs via  Belov codes  based on Güneri's construction}

In this section, we will consider presenting more optimal LCPs by using Belov codes in a modified version of Güneri's construction.


Now we want to delete  $F_1 = W \setminus A$ to construct Belov code, where $ \dim W = \theta $ and  $|A| \in \{\theta, \theta + 1\}$, $|F_1| = |W| - |A| = 2^\theta - 1 - \theta $, or $|F_1| = |W| - |A| = 2^\theta - 1 - (\theta +1)$. 

\begin{lemma}\label{lem8}
Based on the Güneri's construction,  if $h=1,$ then it is not possible to perform the Belov deletion  $F_1 = W \setminus A$ in $[S_k \backslash  U]$, where $U$ is a subspace of $\Bbb F_2^k$ with dimension $k-1$.
\end{lemma}
\begin{proof} Without loss of generality, we can take $U=U_{k-3}$ in Eq. \eqref{eq2}.
Note that $\theta\ge 3$ from Lemma \ref{lem3}.
Let $\{w_1,w_2,w_3\}$ be a set of linearly independent vectors of $W$. 

(i) When all basis vectors are even-weight vectors, let $A =\{w_1,w_2,w_3\}$. Then $F_1 = \{ w_1+w_2,\ w_1+w_3,\ w_2+w_3,\ w_1+w_2+w_3 \}$, it is clear that the sum of any two even-weight vectors is still an even-weight vector, so $F_1 \cap U_{k-3} \neq \emptyset$.

(ii) When $w_1$ and $w_2$ are even-weight vectors and $w_3$ is an odd-weight vector, put all even-weight vectors into $A$, i.e. $A = \{ w_1, w_2, w_1 + w_2 \}$. Although $(W \setminus A) \cap U_{k-3} = \emptyset$, the requirement that any set of $\leq \theta$ points in $A$ are linearly independent is obviously not satisfied.

(iii) When $w_1$ is an even-weight vector and $w_2, w_3$ are odd-weight vectors, put all even-weight vectors into $A$, i.e. $A = \{ w_1, w_2 + w_3, w_1 + w_2 + w_3 \}$. Although $(W \setminus A) \cap U_{k-3} = \emptyset$, the requirement that any set of $\leq \theta$ points in $A$ are linearly independent is still not satisfied.

(iv) When all basis vectors are odd-weight vectors, put all even-weight vectors into $A$, i.e. $A = \{w_1 + w_2, w_1 + w_3, w_2 + w_3 \}$, which is linearly dependent and does not satisfy the condition.

In summary, there does not exist a subspace of $W$ such that $(W\setminus A) \cap U_{k-3} = \emptyset$.
\end{proof}



Due to Lemma \ref{lem8}, we can not perform the Belov deletion when $h=1$ {and $\dim U=k-1$}. Next we always assume that $h\ge 2$ in the rest of this section.  Although we cannot perform the additional deletion of $F_1$ for the Belov code in the first simplex code, we can delete $F_1$ in subsequent simplex codes such that $F_1 \cap U = \emptyset$ for some subspace $U$ whose dimension is less than $k-1$. To satisfy the Griesmer bound, only one copy of $F_1$ can be deleted from the simplex code, and we choose to perform the deletion in $U_{k-4}$ with dimension $k-2$.

In the construction method of Güneri with $u_1=3$, we construct two classes of optimal binary LCP codes from Belov codes rather than SS codes in the following theorem.

\begin{theorem}\label{lem9}
Let \(k\ge 2\), \(h\ge 2\), and let
$1\le u_1<u_2<\cdots<u_t<k$
satisfy
\[
\sum_{i=1}^{\min\{h+1,t\}}u_{t+1-i}\le hk.
\]
Assume that the corresponding Solomon--Stiffler code has minimum distance
$
d_{\mathrm{SS}}=h2^{k-1}-\sum_{i=1}^{t}2^{u_i-1},
$
and that
$
d_{\mathrm{SS}}
\ge
\left\lceil\frac{k-1}{2}\right\rceil 2^{k-1}.
$
Consider the maximal anticode case, so that $u_1=3$. Let $\theta=3$, { and since $3=\theta\le u_1 < k$, we have $k \ge 4$.} Choose $A\subset W$ with $|A|\in\{3,4\}$, and put
$
\varepsilon=\theta+1-|A|\in\{0,1\}.
$
Assume that the corresponding Belov code has minimum distance by
\[
d_{\mathrm{Belov}}
=
d_{\mathrm{SS}}-\sum_{i=1}^{\theta-2}2^i-\varepsilon
=
d_{\mathrm{SS}}-2-\varepsilon.
\]
Then there exists an optimal binary LCP \((\mathcal C,\mathcal D)\) of length
$
g_2(k,d_{\mathrm{Belov}})
$
and dimension \(k\), with security parameter \(d_{\mathrm{Belov}}\), such that \(\mathcal C\) is a Belov code with parameters
$
[g_2(k,d_{\mathrm{Belov}}),\,k,\,d_{\mathrm{Belov}}].
$
\end{theorem}

\begin{proof}
In this proof, matrix addition and multiplication are performed over $\mathbb{F}_2$.
We want to delete more columns   to use Belov construction, which means $u_1\ge 3$ from Lemma \ref{lem3}.  

Firstly, from our assumptions, we have $h\ge 2$ and let $G=[S_{k}(h) \setminus F]=[I_k:M].$ Let  $n$ be the length of the code and then \[
M = (m_1, m_2, \ldots, m_{n-k}), \quad \text{where } m_i \in \mathbb{F}_2^{k\times 1} \text{ for } 1 \le i \le n-k.
\] Then we have $
MM^T 
= I_k,
$  when $u_1\ge 3$ in Güneri's construction, meaning that any two distinct rows of $M$ have even intersection weight, while each row has odd weight. 
However,  if we remove $F_1$ in $M$, then the equality $MM^T = I_k$ no longer holds. 

Further deleting \(F_1 = W \setminus A\) yields the Belov code. For the resulting matrix $M'$, we have
\[
M'M'^T
= \sum_{i \notin F_1} m_i m_i^T
= \left( \sum_{i=1}^{n-k} m_i m_i^T \right) + \sum_{i \in F_1} m_i m_i^T.
\]
Let $\Delta = \sum_{i \in F_1} m_i m_i^T$. Then
\[
M'M'^T = MM^T + \Delta = I_k + \Delta.
\]

In the maximal anticode case, we have $u_1 = 3 \ge \theta \ge 3$, which forces $\theta = 3$. Thus, $\dim W = \theta = 3$ and $|W| = 2^\theta $.

(i). If   $\varepsilon = \theta + 1 - |A| = 1$ for $|A| = \theta $,  then we choose a basis for $W$ as
\[
w_1 = \begin{pmatrix}0\\0\\ \vdots \\0 \\0 \\ 1\end{pmatrix},\quad
w_2 = \begin{pmatrix}0\\ \vdots\\0 \\1 \\1 \\ 0\end{pmatrix},\quad
w_3 = \begin{pmatrix}0\\1\\0\\ \vdots\\0 \\0 \end{pmatrix}
\hspace{-0.4em}
\raisebox{-1em}{$\left.\rule{0pt}{2.5em}\right\}$}
\hspace{0.3em}
\raisebox{-1em}{$\text{all-zero } (k-3)\text{ rows.}$}
\]
In order to satisfy the condition that any set of at most $ \theta$ vectors in $A$ are linearly independent, we choose the basis of $W$ to form the set $A$. 
Accordingly, we obtain $F_1$ as shown below:
\[
F_1 =
\left(
\begin{matrix}
0 & 0 & 0 & 0 \\
0 & 1 & 1 & 1 \\
0 & 0 & 0 & 0 \\
\vdots & \vdots & \vdots & \vdots \\
0 & 0 & 0 & 0 \\
1 & 0 & 1 & 1 \\
1 & 0 & 1 & 1 \\
1 & 1 & 0 & 1
\end{matrix}
\right)
\hspace{-0.4em}
\raisebox{0.8em}{$\left.\rule{0pt}{2.5em}\right\}$}
\hspace{0.3em}
\raisebox{0.8em}{$\text{all-zero } (k-5)\text{ rows}.$}
\] Then we have  
\[
\Delta = \sum_{i \in F_1} m_i m_i^T = F_1F_1^T =
\left(
\begin{matrix}
0 & 0 &  \\ 
0 & 1 &  \\ 
 &  & 0 & \cdots &0\\
 & &\vdots & \ddots  & \vdots \\
 &  & 0 & \cdots & 0  \\
 &  & & &  & 1 & 1 & 0\\
 &  & & &  & 1 & 1 & 0\\
 &  & & &  & 0 & 0 & 1 
\end{matrix}
\right)
\hspace{-0.4em}
\raisebox{0.8em}{$\left.\rule{0pt}{2.5em}\right\}$}
\hspace{0.3em}
\raisebox{0.8em}{$\text{all-zero } (k-5)\text{ rows}$}.
\]

Then 
\[
M'M'^T =  I_k + \Delta =
\begin{pmatrix}
1 & 0 & \\
0 & 0 & \\
 & & \smash{\raisebox{3pt}{$I_{k-5}$}} & & & &\\
 & & & 0 & 1 & 0 \\
 & & & 1 & 0 & 0 \\
 & & & 0 & 0 & 0 \\
\end{pmatrix}. 
\]
Let \(x_3 = (1, 1, \dots, 1,0,1,1)^T\) and \(x_4 = (0, 1, \dots, 1,0,1,1)^T\) be vectors of length \(n\) over \(\mathbb{F}_2\). We denote by \(\sigma_1(M')\) the matrix resulting from swapping \(x_3\) and \(x_4\) in \(M'\) and by
\[
\varphi_1 = x_3x_3^T + x_4x_4^T + x_3x_4^T + x_4x_3^T= \begin{pmatrix}
1 & 0 & \cdots & 0 \\
0 & 0 & \cdots & 0 \\
\vdots & \vdots &\ddots &\vdots\\
0 & 0 & \cdots & 0 \\
\end{pmatrix}.
\] Therefore 
$M'\sigma_1(M')^T = M'M'^T + \varphi_1$, and let us cyclically shift row of $\sigma_1(M')$ to obtain the matrix $N'^T$ as follows:
\[
N'^T =P_1\sigma_1(M') =\begin{pmatrix}
\sigma_1(M')_3 \\
\sigma_1(M')_4 \\
\vdots \\
\sigma_1(M')_k \\
\sigma_1(M')_1 \\
\sigma_2(M')_2
\end{pmatrix}, 
~
\mbox{ where~} 
P_1 = \begin{pmatrix}
0 & 0 & 1 &\cdots&0\\
\vdots &\vdots &\vdots&\ddots&\vdots \\
0 & 0 & 0 &\cdots&1\\
1 & 0 & 0 &\cdots&0\\
0 & 1 & 0 &\cdots&0
\end{pmatrix}.
\]

It follows {that} $M'N' = M'(N'^T)^T = M'(P_1\sigma_1(M'))^T = M'\sigma_1(M')^TP_1^T$. Notice that the matrix $M'N'$  can be obtained by cyclically shifting the columns of $M'\sigma_1(M')^T$ as follows:
\[
M'N' = 
\begin{pmatrix}
0 \\
0 & 0 \\
1&\cdots & 0\\
\vdots &\ddots &\vdots&\ddots \\
0 & \cdots & 1 &0 &0\\
0& \cdots& 0& 0 &1&0&\\
0& \cdots& 0&1&0 &0& 0 \\
0& \cdots& 0&0 &0 &0 & 0 &0 \\
\end{pmatrix}.
\]
Since \(M'N'\) is a lower triangular matrix, all diagonal entries of \(I_k + M'N'\) are equal to one and hence nonzero, which implies that \(I_k + M'N'\) is invertible. {By the  Güneri's argument in Section III, our obtained code is optimal.}

(ii). If  $\varepsilon = \theta + 1 - |A| = 0$ for $|A| = \theta +1$,  then we choose a basis for $W$ as
\[
w_1 = \begin{pmatrix}0\\0\\0\\0\\0\\ \vdots \\0 \\ 1\end{pmatrix},\quad
w_2 = \begin{pmatrix}0\\1\\1\\1\\0\\ \vdots\\0 \\ 0\end{pmatrix},\quad
w_3 = \begin{pmatrix}1\\1\\0\\0\\0\\ \vdots\\0 \\0 
\end{pmatrix}
\hspace{-0.4em}
\raisebox{-2.7em}{$\left.\rule{0pt}{2.5em}\right\}$}
\hspace{0.3em}
\raisebox{-2.7em}{$\text{all-zero } (k-5)\text{ rows.}$}
\]
Then $F_1$ can be shown below:
\[
F_1 = 
\begin{pmatrix}
0 & 1 & 1 \\
1 & 0 & 1 \\
1 & 1 & 0 \\
1 & 1 & 0 \\
0 & 0 & 0 \\
\vdots & \vdots & \vdots  \\
0 & 0 & 0
\end{pmatrix}
\hspace{-0.4em}
\raisebox{-3.5em}{$\left.\rule{0pt}{2.5em}\right\}$}
\hspace{0.3em}
\raisebox{-3.5em}{$\text{all-zero } (k-4)\text{ rows.}$}
\] Hence we have 
\[
\Delta = F_1F_1^\mathrm{T} =
\begin{pmatrix}
0 & 1 & 1 & 1 & 0 & \cdots & 0\\ 
1 & 0 & 1 & 1 & 0 & \cdots & 0\\ 
1 & 1 & 0 & 0 & 0 & \cdots & 0\\
1 & 1 & 0 & 0 & 0 & \cdots & 0 \\
0 & 0 & 0 & 0 & 0 & \cdots & 0 \\
\vdots & \vdots & \vdots& \vdots& \vdots& \ddots  & \vdots \\
0 & 0 & 0 & 0 & 0 & \cdots & 0 
\end{pmatrix}
\hspace{-0.4em}
\raisebox{-3.5em}{$\left.\rule{0pt}{2.5em}\right\}$}
\hspace{0.3em}
\raisebox{-3.5em}{$\text{all-zero } (k-4)\text{ rows,}$}
\]
and
\[
M'M'^T =  MM^T + \Delta =
\begin{pmatrix}
1 & 1 & 1 & 1 &  \\ 
1 & 1 & 1 & 1 &   \\ 
1 & 1 & 1 & 0 &   \\
1 & 1 & 0 & 1 &   \\
  &   &   &   & I_{k-4}  \\
\end{pmatrix}.
\]
Let \(x_5 = (0, 0, 1, 1, 0, \dots,0, 1)^T\) and \(x_6 = (1, 1,0, 0,0, \dots, 0, 1)^T\).  We denote by \(\sigma_2(M')\) the matrix resulting from swapping \(x_5\) and \(x_6\) in \(M'\) and  
\[
\varphi_2 = x_5x_5^T + x_6x_6^T +x_5x_6^T + x_6x_5^T= \begin{pmatrix}
1 & 1 & 1 & 1 & 0 & \cdots & 0\\ 
1 & 1 & 1 & 1 & 0 & \cdots & 0\\ 
1 & 1 & 1 & 1 & 0 & \cdots & 0\\
1 & 1 & 1 & 1 & 0 & \cdots & 0 \\
0 & 0 & 0 & 0 & 0 & \cdots & 0 \\
\vdots & \vdots & \vdots& \vdots& \vdots& \ddots  & \vdots \\
0 & 0 & 0 & 0 & 0 & \cdots & 0 
\end{pmatrix}.
\] Then
\[M'\sigma_2(M')^T = M'M'^T + \varphi_2 =
\begin{pmatrix}
0 & 0 & 0 & 0 & \\ 
0 & 0 & 0 & 0 & \\ 
0 & 0 & 0 & 1 &  \\
0 & 0 & 1 & 0 &   \\
 &  &  & & I_{k-4}  \\
\end{pmatrix}.
\]
Let us cyclically shift row of $\sigma_2(M')$ to obtain the matrix $N'^T$ as follows:
\[
N'^T=P_1\sigma_2(M')= \begin{pmatrix}
\sigma_2(M')_3 \\
\sigma_2(M')_4 \\
\vdots \\
\sigma_2(M')_k \\
\sigma_2(M')_1 \\
\sigma_2(M')_2
\end{pmatrix},
~
\mbox{ where~} 
P_1 = \begin{pmatrix}
0 & 0 & 1 &\cdots&0\\
\vdots &\vdots &\vdots&\ddots&\vdots \\
0 & 0 & 0 &\cdots&1\\
1 & 0 & 0 &\cdots&0\\
0 & 1 & 0 &\cdots&0
\end{pmatrix}.
\]
Namely, the matrix $M'N'$ is obtained by cyclically permuting the columns of $M'\sigma_2(M')^\mathrm{T}$:
\[
M'N' = 
\begin{pmatrix}
0 \\
0 & 0 \\
0&1 & 0\\
1& 0& 0&0 \\
0& 0& 1& \cdots&0 \\
\vdots &\vdots &\vdots& \ddots&\vdots& \ddots\\
0 & 0 & 0&\cdots&1 &0 &0 \\
\end{pmatrix}.
\]
It is easy to check that \(I_k + M'N'\) is invertible.   By the  Güneri's argument in Section III, the resulting code is optimal. This completes the proof.
\end{proof}



 {\rm




{ It is important comparing $
[g_2(k,d_{\mathrm{SS}}),\,k,\,d_{\mathrm{SS}}]
$ by Güneri's construction of maximal anticode cases with $
[g_2(k,d_{\mathrm{Belov}}),\,k,\,d_{\mathrm{Belov}}]$ by our construction when $k=4,5,6,7,8$.} The main purpose of Table~\ref{table:parameters-comparison} is to compare the minimum distance bounds of Belov construction with Güneri's previous Solomon-Stiffler (SS) construction across different dimensions $k$.
In Güneri's construction, when restricted to $u_1 \geq 3$, the minimum distance $d$ is required to be large (e.g., $d \geq 36$ for $k=5$), which severely limits the range of available parameters. In contrast, Belov construction successfully lowers the minimum distance lower bound under the same condition of $u_1 \geq 3$ (e.g., reduced to $\geq 33$ for $k=5$). 
This Table~\ref{table:parameters-comparison} demonstrates that Belov codes not only cover the parameter regime of SS codes but also broaden the existence range of optimal LCPs, constructing LCPs with smaller distances that still achieve the Griesmer bound. It verifies the core contribution mentioned in our conclusion: relaxing the restriction of $d \geq \lceil (k-1)/2 \rceil 2^{k-1}$.}

\begin{table}[!htbp] 
\begin{center} 
\caption{Minimum Distance Lower Bounds of LCPs from  Solomon-Stiffler and Belov Codes}
\label{table:parameters-comparison}
\begin{tabu}{|c|c|c|c|c|}
\hline
\multirow{2}{*}{} & \multicolumn{3}{c|}{Solomon-Stiffler Code \cite{G}} & Belov Code \\
\hline
&$u_1 \ge 3$ & $u_1=2$ & $u_1=1$ & $u_1 \ge 3, \theta = 3$  \\
\hline
$k=4$& & &  & $ \ge9$  \\
\hline
$k=5$&$ \ge 36$ & $\ge34$ & $ \ge 33$ & $ \ge33$  \\
\hline
$k=6$&$ \ge100 $ & $ \ge98$ & $ \ge 97$ & $ \ge97$  \\
\hline
$k=7$&$ \ge196 $ & $ \ge194$ & $\ge193$ & $ \ge193$  \\
\hline
$k=8$&$ \ge516 $ & $\ge514$ & $ \ge513$ & $ \ge513$  \\
\hline
\end{tabu}
\end{center} 
\end{table}

}

 \section{More new optimal binary LCPs}

This section is devoted to constructing further binary optimal LCPs by employing SS codes and Belov codes with different integers  $h$.

First, we present two following examples by Magma, in cases of $h=1$ and $h=2$. 

\begin{example}\label{belov} 
{\rm
Let $h=1$, $k = 6$, $d = 21 $, and $g_2(k, d) = 44$. We can delete a 4-dimensional  subspace $U$ wohse basis is  $\mathcal{B} \setminus  \{ \left( {0,1,0,\ldots ,0,1}\right) \},$ {where $\mathcal{B}$  was defined  in Equation (5).} Choose other parameters as follows: {Let} $|A| = \theta = 3 $. We have $\varepsilon = \theta + 1 - |A| = 1 $.  Choose a basis for $W$ as 
$w_1 = (0,0,0,0,0,1)^T, w_2=(0,0,0,1,1,0)^T, w_3=(0,1,0,0,0,0)^T$. The number of columns deleted  is $|F_1|=|W \setminus A| = 4$, {so that} $F_1=W \setminus A$ can be {written as}
\[
F_1= 
\begin{pmatrix}
0&0&0&0 \\
0&1&1&1\\
0&0&0&0\\
1&0&1&1 \\
1&0&1&1 \\
1&1&0&1
\end{pmatrix}.
\]
In this case, we can obtain a Belov code with parameters $[44, 6, 21]$. We have
\[
MM^T=\begin{pmatrix}
1&0&0&0&0&0\\ 
0&0&0&0&0&0\\
0&0&1&0&0&0\\
0&0&0&0&1&0\\ 
0&0&0&1&0&0\\ 
0&0&0&0&0&0 
\end{pmatrix},
\]
we swap columns $x_1 = (1, 1, 1,1, 1, 1)^T$ and $x_2 = (0, 1, 1,1, 1, 1)^T$ and denote by \(\sigma_3(M)\) the matrix resulting from swapping \(x_1\) and \(x_2\). The matrix $N^T$ is constructed as follows:
\[
N^T =P_1\sigma_3(M) =\begin{pmatrix}
\sigma_3(M)_3 \\
\sigma_3(M)_4 \\
\vdots \\
\sigma_3(M)_k \\
\sigma_3(M)_1 \\
\sigma_3(M)_2
\end{pmatrix}, 
~
\mbox{ where~} 
P_1 = \begin{pmatrix}
0 & 0 & 1 &\cdots&0\\
\vdots &\vdots &\vdots&\ddots&\vdots \\
0 & 0 & 0 &\cdots&1\\
1 & 0 & 0 &\cdots&0\\
0 & 1 & 0 &\cdots&0
\end{pmatrix}.
\]
This cyclic shift operation ensures that $\D^\perp$ maintains the same parameters as $\C$. The resulting matrix $MN$ is strictly lower triangular:
\[
MN= \begin{pmatrix} 
0\\ 
0&0\\
1&0&0\\
0&0&1&0\\ 
0&1&0&0&0\\ 
0&0&0&0&0&0 
\end{pmatrix}.
\]
{It is easy to see the matrix $I_k + MN$ must be invertible. Since both $\C$ and $\D^\perp$ have parameters $[44, 6, 21]$, the security parameter is
$
d_{LCP} = \max \{ \min\{d(\C), d(\D^\perp)\} \} = 21.
$}
}
\end{example}

\begin{example}\label{ex:LCPs-construction} 
{\rm
Let $h=2$, $k = 6$, $d = 37$, and $g_2(k, d) = 76$. Since $h=2$, we can delete a 4-dimensional subspace $U_1$ and 5-dimensional subspace $U_2$, and their bases are listed as follows: 
$\mathcal{B} \setminus \{ \left( {0,1,0,\ldots ,0,1}\right) \} $, where $\mathcal{B}$  was defined  in Equation (5).


Choose a basis for $W$ as 
$w_1 = (0,1,0,0,0,0)^T, w_2=(0,0,1,0,0,0)^T, w_3=(0,0,0,0,1,1)^T$. The number of columns deleted  is $|F_1|=|W \setminus A| = 4$, and so $F_1=W \setminus A$ can be shown below:
\[
F_1= 
\begin{pmatrix}
0&0&0&0 \\
1&1&0&1\\
1&0&1&1\\
0&0&0&0 \\
0&1&1&1 \\
0&1&1&1
\end{pmatrix}.
\]
In this case, we can obtain a Belov code with parameters $[76, 6, 37]$.  Then we have
\[
MM^T=\begin{pmatrix}
1&0&0&0&0&0\\ 
0&0&0&0&0&0\\
0&0&0&0&0&0\\
0&0&0&1&0&0\\ 
0&0&0&0&0&1\\ 
0&0&0&0&1&0
\end{pmatrix},
\]
we swap columns $x_1 = (1, 1, 1,1, 1, 1)^T$ and $x_2 = (0, 1, 1,1, 1, 1)^T$ and denote by \(\sigma_3(M)\) the matrix resulting from swapping \(x_1\) and \(x_2\). The matrix $N^T$ is constructed as follows:
\[
N^T = P_1\sigma_3(M) =  \begin{pmatrix}
\sigma_3(M)_3 \\
\sigma_3(M)_4 \\
\vdots \\
\sigma_3(M)_k \\
\sigma_3(M)_1  \\
\sigma_3(M)_2
\end{pmatrix},
~
\mbox{ where~} 
P_1 = \begin{pmatrix}
0 & 0 & 1 &\cdots&0\\
\vdots &\vdots &\vdots&\ddots&\vdots \\
0 & 0 & 0 &\cdots&1\\
1 & 0 & 0 &\cdots&0\\
0 & 1 & 0 &\cdots&0
\end{pmatrix}.
\]
The resulting matrix $MN$ is strictly lower triangular:
\[
MN= \begin{pmatrix} 
0\\ 
0&0\\
0&0&0\\
0&1&0&0\\ 
0&0&0&1&0\\ 
0&0&1&0&0&0 
\end{pmatrix}.
\]
It is easy to see the matrix $I_k + MN$ must be invertible. Since both $\C$ and $\D^\perp$ have parameters $[76, 6, 37]$, the security parameter is
$
d_{LCP} = \max \{ \min\{d(\C), d(\D^\perp)\} \} = 37.
$
}
\end{example}

Inspired by the above two examples, now we will try to give some constructions of LCPs by using SS codes and Belov codes according to the value of $h$. 


\begin{lemma}\label{lem12}
    Assume that $k \ge 2$ and $\lceil \frac{k-1}{2} \rceil2^{k-1}\ge d = h2^{k - 1} - \sum_{i = 1}^{t} 2^{u_i - 1}(d \neq 2^{k-2})$ with  $1 \leq u_1 < u_2 < \cdots < u_t < k$ and $\sum_{i=1}^{{min}\{h+1,t\}}u_{t+1-i}\le hk $.
    Let $F = \cup_{i = 1}^{t} U_i   $ with $\mbox{dim}(U_i)=u_i$ and $F_1 = W \setminus A$ for some subspace $W$ such that $F\cap F_1=\emptyset$.  Then we can find a generator matrix $G = [{S}_{k} (h)  \setminus  F]$ or $ [S_k(h) \setminus (F \cup F_1)]$ of \(\C\) with the following properties:
\begin{enumerate}[label=\roman*.]
    \item[(i)] If $h=1$, i.e. $d \le 2^{k-1}$, then $G$ has

    \begin{enumerate}[label=(\arabic*)]
        \item the $k \times k$ identity matrix $I_k$ and
        \item the following column vectors: \\
        $x_3 = (1,1,\cdots,1,0,1,1)^T$, 
        $x_4 = (0,1,\cdots,1,0,1,1)^T$, \\
        $x_5 = (0, 0, 1, 1, 0, \dots,0, 1)^T$, 
        $x_6 = (1, 1,0, 0,0, \dots, 0, 1)^T$.
    
    \end{enumerate}
    \item[(ii)] If $2 \le h \le \lceil \frac{k-1}{2} \rceil$, i.e. $ 2^{k-1} < d \le \lceil \frac{k-1}{2} \rceil 2^{k - 1}$, then $G$ has
    \begin{enumerate}[label=(\arabic*)]
        \item at least one copy of the  $k \times k$ identity matrix $I_k$ and
        \item at least one copy of the following column vectors:\\
        $x_1 = (1,1,\cdots,1,1)^T$,
        $x_2 = (0,1,\cdots,1,1)^T$,\\
        $x_3 = (1,1,\cdots,1,0,1,1)^T$, 
        $x_4 = (0,1,\cdots,1,0,1,1)^T$, \\
        $x_7 = (1,1,0,\cdots,0, 1)^T$, 
        $x_8 = (0,0,0,\cdots,0,1)^T$. 
    \end{enumerate}
    
\end{enumerate}
\end{lemma}

\begin{proof} We first construct the Solomon–Stiffler code, and then we obtain the Belov code by further deleting $F_1$ from the SS codes.

    \begin{enumerate}[label=\roman*.]
        \item[(i).] $h=1$.          
        \begin{enumerate}[label=(\arabic*)]
            \item (SS code)  By Lemma \ref{lem8}, $u_t\neq k-1$. To keep the identity $I_k$ in $G$, we have to delete the subspaces of $U_{k-3}$ in Eq. \eqref{eq2}. Consider the maximal anticode case, we can delete  the space spanned by the set $\mathcal{B} \setminus \{ \left( {0,1,0,\ldots ,0,1}\right) \}$. This deletion of maximal subspace ensures the existence of $I_k$, $x_3$, $x_4$, $x_5$ and $x_6$.
            \item (Belov code) Choose $\theta = 3$, $\text{dim}(W)=\theta$.

            \begin{enumerate}[label=\arabic*)]
                \item $\varepsilon = \theta +1-|A| =1$ for $|A|=\theta$, then we choose a basis for $W$
                {
                 in the same manner as part (i) in the proof of Theorem 9, ensures the existence of $x_3$, $x_4$.
                }
                
                \item $\varepsilon = \theta +1-|A| =0$ for $|A|=\theta+1$, 
                 then we choose a basis for $W$ 
                 {
                 in the same manner as part (ii) in the proof of Theorem 9, ensures the existence of $x_5$, $x_6$.
                    }
            \end{enumerate}
        \end{enumerate}
        
        \item[(ii).] $2 \le h \le \lceil \frac{k-1}{2} \rceil$.
            \begin{enumerate}[label=(\arabic*)]
                \item (SS code)
                Let us consider the \lq\lq maximal anticode\rq\rq  case for $u_1 \geq 3$ by deleting subspaces $U_1,  \dots, U_{k-3} \subset {\mathbb{F}}_{2}^{k} $ with dimensions 
                \[
                u_1 = 3, u_2=4, \dots, u_{k-4}=k-2, u_{k-3}=k-1
                \]
                and $\sum_{i=1}^{{min}\{h+1,t\}}u_{t+1-i}\le hk $.
                
                We define the subspace $U_{k-3}$ of the largest dimension, and $\mathcal{B}$ is a basis for $U_{k-3}$. The next largest subspace $U_{k-4}$ of dimension $k-2$ as the space spanned by the set $\mathcal{B} \setminus \{ \left( {0,1,0,\ldots ,0,1}\right) \}$. When $h=2$, the maximal anticode is $u_1=k-2, u_2=k-1$, it is clear that $I_k \notin U_2$, and ensures $\{x_1,x_2,x_3, x_4, x_7,x_8\} \notin U_1$, and so on.

                \item (Belov code) To satisfy the Griesmer code, only one copy of $F_1$ can be deleted from the simplex code, and we choose to perform the deletion in $U$ of dimension $k-2$ for the maximal anticode case.
                \begin{enumerate}[label=\arabic*)]
                \item $\varepsilon = \theta +1-|A| =1$ for $|A|=\theta$, We choose a basis for $W$ as
                \[
                w_1 = \begin{pmatrix}0\\1 \\0 \\0 \\ \vdots \\ 0 \\0\end{pmatrix},\quad
                w_2 = \begin{pmatrix}0\\0 \\1 \\0 \\  \vdots \\ 0 \\0\end{pmatrix},\quad
                w_3 = \begin{pmatrix}0 \\0\\0\\ \vdots \\0\\ 1 \\1\end{pmatrix}
                \hspace{-0.4em}
                \raisebox{0em}{$\left.\rule{0pt}{2.5em}\right\}$}
                \hspace{0.3em}
                \raisebox{0em}{$\text{all-zero } (k - 4) \text{ rows}$}.
                \] Then 
                $F_1$ can be written as shown below:
                    \[
                    F_1 = 
                    \begin{pmatrix}
                    0 & 0 & 0 & 0\\
                    1 & 1 & 0 & 1 \\
                    1 & 0 & 1 & 1\\
                    0 & 0 & 0 & 0\\
                    \vdots & \vdots & \vdots  & \vdots\\
                    0 & 0 & 0 & 0 \\
                    0 & 1 & 1 & 1\\
                    0 & 1 & 1 & 1\\
                    \end{pmatrix}
                    \hspace{-0.4em}
                    \raisebox{-1em}{$\left.\rule{0pt}{2.5em}\right\}$}
                    \hspace{0.3em}
                    \raisebox{-1 em}{$\text{all-zero } (k-5)\text{ rows.}$}
                    \]
                    ensures the existence of $x_3$ , $x_4$ for $k \neq 6$. Because $x_4=(0,1,1,0,1,1)^T \in F_1$ when $k=6$ , ensures the existence of $x_1=(1,1,1,1,1,1)^T$ , $x_2=(0,1,1,1,1,1)^T$.  
                
                \item $\varepsilon = \theta +1-|A| =0$ for $|A|=\theta+1$, 
                We choose a basis for $W$ as
                \[
                w_1 = \begin{pmatrix}1 \\ 0 \\ 0\\ \vdots\\0\\ 0 \\0\end{pmatrix},\quad
                w_2 = \begin{pmatrix}0 \\ 1 \\ 0\\ \vdots\\0\\ 0 \\0\end{pmatrix},\quad
                w_3 = \begin{pmatrix}0 \\ 0 \\ 1\\0\\ \vdots \\ 0\\1\end{pmatrix}
                \hspace{-0.4em}
                \raisebox{-1.8em}{$\left.\rule{0pt}{2.5em}\right\}$}
                \hspace{0.3em}
                \raisebox{-1.8em}{$\text{all-zero } (k - 4) \text{ rows}$.}.
                \] Then 
                $F_1$ can be written as shown below:
                    \[
                    F_1 = 
                    \begin{pmatrix}
                    1 & 0 & 1  \\
                    0 & 1 & 1 \\
                    0 & 0 & 0 \\
                    \vdots & \vdots & \vdots \\
                    0 & 0 & 0  \\
                    0 & 0 & 0 \\
                    0 & 0 & 0 \\
                    \end{pmatrix}
                    \hspace{-0.4em}
                    \raisebox{0em}{$\left.\rule{0pt}{2.5em}\right\}$}
                    \hspace{0.3em}
                    \raisebox{0em}{$\text{all-zero } (k-4)\text{ rows.}$}
                    \]
                ensures the existence of $x_7$ , $x_8$.
            \end{enumerate}
            \end{enumerate}
    \end{enumerate}

    This completes the proof.
\end{proof}

\begin{remark}
{
    {\rm
In Lemma \ref{lem12} (i), if $h=1$, then $d\le 2^{k-1}$ which is not covered by Güneri's construction.
In Lemma \ref{lem12} (ii), if $2\le h\le \left\lceil \frac{k-1}{2}\right\rceil$, then $2^{k-1}<d\le \left\lceil \frac{k-1}{2}\right\rceil 2^{k-1}$, which is the range below Güneri's bound assumption for $d$.

}
}
\end{remark}

{
In Theorem \ref{lem9}, for $h\ge 2$ we provided an explicit optimal binary LCPs from Belov construction. By using  Lemma \ref{lem12}, we cover the remaining cases $h=1$ and $2\le h\le \lceil (k-1)/2\rceil$ in the following theorem. }


\begin{theorem}\label{sun-newthm1}
 Assume that $k\ge 2$, $d = h2^{k - 1} - \sum_{i = 1}^{t} 2^{u_i - 1}$ with $3\leq u_1 < u_2 < \cdots < u_t < k$, and $\sum_{i=1}^{{min}\{h+1,t\}}u_{t+1-i}\le hk $.
     \begin{enumerate}[label=\roman*.]
         \item[(i)] When $h=1$, i.e. $d \le 2^{k-1}$, if $(t,u_1)\neq (1,k-1)$, then there exists  {an optimal} binary $[g_2(k,d), k]$ LCP  of Solomon–Stiffler or Belov type  with security parameters $d$.

         \item[(ii)] When $2 \le h \le \lceil \frac{k-1}{2} \rceil$, i.e. $ 2^{k-1} < d \le \lceil \frac{k-1}{2} \rceil 2^{k - 1}$, there exists   an optimal binary $[g_2(k,d), k]$ LCP of Solomon–Stiffler or Belov type  with security parameters $d$.
     \end{enumerate}
\end{theorem}

{
\begin{proof}
We begin by constructing a Belov code $\C$ with parameters $[g_2(k,d), k, d]$, which satisfies the Griesmer bound. 
 According to Lemma \ref{lem4}, such a code exists when we choose appropriate parameters $\theta$, $A$, and subspaces $U_i$ satisfying the required conditions.

Let $G = [S_k(h) \setminus (F \cup F_1)]$ be a generator matrix of $\C$, where $F = \bigcup_{i=1}^t U_i$ is the union of disjoint subspaces for SS codes and $F_1 = W \setminus A$ is the additional puncturing for the Belov construction. Through suitable column permutations, we can transform $G$ into systematic form $G = [I_k : M']$, where $M'$ is a $k \times (n-k)$ matrix.

The core of the proof lies in analyzing the structure of $M'M'^T$ and constructing an appropriate matrix $N'$ such that $I_k + M'N'$ is invertible. 

Notice that if $h\ge  \lceil \frac{k-1}{2} \rceil$, then the result follows from \cite[Theorem 5]{G23} for SS codes and Theorem 10 for Belov codes. Next we proceed the proof by the rest two cases.

\begin{enumerate}[label=\roman*.]
    \item[(i).] $h=1$.          
    \begin{enumerate}[label=(\arabic*)]
        \item (SS code) By Lemma \ref{lem12}, we can consider the maximal anticode case with $t=1$ and $u_1=k-2$. Then take $F_1=\emptyset $ and we have $M'M'^T=I_k$ from  \cite[Theorem 4]{G23}.  Then the result follows from \cite[Theorem 5]{G23}.
        
        \item (Belov code) By Lemma \ref{lem12}, choose  $\text{dim}(W)=\theta= 3$. Let $\Delta = \sum_{i \in F_1} m_i m_i^T = F_1F_1^T$. Then $M'M'^T = MM^T + \Delta = I_k + \Delta.$
        \begin{enumerate}[label=\arabic*)]
            \item If $\varepsilon = \theta +1-|A| =1$ for $|A|=\theta$, 
            { 
            then the argument follows identically to part (i) in the proof of Theorem 9 and we omit it here.
            We have 
            \[
            M'N' = 
            \begin{pmatrix}
            0 \\
            0 & 0 \\
            1&\cdots & 0\\
            \vdots &\ddots &\vdots&\ddots \\
            0 & \cdots & 1 &0 &0\\
            0& \cdots& 0& 0 &1&0&\\
            0& \cdots& 0&1&0 &0& 0 \\
            0& \cdots& 0&0 &0 &0 & 0 &0 \\
            \end{pmatrix}.
            \]
            Since \(M'N'\) is a lower triangular matrix, which implies that \(I_k + M'N'\) is invertible.
            }

            \item If $\varepsilon = \theta +1-|A| =0$ for $|A|=\theta+1$, 
            {
            then the argument follows identically to part (ii) in the proof of Theorem 9 and we omit it here.
            Then we have
            \[
            M'N' = 
            \begin{pmatrix}
            0 \\
            0 & 0 \\
            0&1 & 0\\
            1& 0& 0&0 \\
            0& 0& 1& \cdots&0 \\
            \vdots &\vdots &\vdots& \ddots&\vdots& \ddots\\
            0 & 0 & 0&\cdots&1 &0 &0 \\
            \end{pmatrix}.
            \]
            It is easy to check that \(I_k + M'N'\) is invertible.
            }
            
        \end{enumerate}
    \end{enumerate}
        
    \item[(ii).] $2 \le h \le \lceil \frac{k-1}{2} \rceil$.
        \begin{enumerate}[label=(\arabic*)]
            \item (SS code) By Lemma \ref{lem12}, we can consider the maximal anticode case with $u_1\ge3$. Then take $F_1=\emptyset $ and we have $M'M'^T=I_k$ from  \cite[Theorem 4]{G23}.  Then the result follows from \cite[Theorem 5]{G23}.
            
            \item (Belov code)            
            \begin{enumerate}[label=\arabic*)]
            \item If $\varepsilon = \theta +1-|A| =1$ for $|A|=\theta$, then we have 
            \[
            \Delta = F_1F_1^T =
            \left(
            \begin{matrix}
            0 & 0 &  0 \\ 
            0 & 1 &  0\\ 
            0 & 0 &  1\\
             &  & &  0 & \cdots &0\\
             & &&\vdots & \ddots  & \vdots \\
             &  && 0 & \cdots & 0  \\
             &  & & &  &  & 1 & 1\\
             &  & & &  &  & 1 & 1\\
            \end{matrix}
            \right)
            \hspace{-0.4em}
            \raisebox{-0.8em}{$\left.\rule{0pt}{2.5em}\right\}$}
            \hspace{0.3em}
            \raisebox{-0.8em}{$\text{all-zero } (k-5)\text{ rows}$}
            \]
            and 
            \[
            M'M'^T =  MM^T + \Delta =
            \begin{pmatrix}
            1 & 0 &0 \\
            0 & 0 &0 \\
            0 & 0 &0 \\
             & && \smash{\raisebox{3pt}{$I_{k-5}$}} & & \\
             & & && 0 & 1  \\
             & & && 1 & 0  \\
            \end{pmatrix}.
            \]
            Let $x_3 = (1,1,\cdots,1,0,1,1)^T$ and $x_4 = (0,1,\cdots,1,0,1,1)^T$ for $k \neq 6$. When $k=6$, the vector $x_4=(0,1,1,0,1,1)^T \in F_1$ will be Belov deleted and in this case we can choose  $x_1=(1,1,1,1,1,1)^T$ , $x_2=(0,1,1,1,1,1)^T$ in the generator matrix.  We denote by \(\sigma_3(M')\) the matrix resulting from swapping \(x_3\), \(x_4\) or $x_1$, $x_2$ in \(M'\). Let   
            \[
            \varphi_1 = x_3x_3^T + x_4x_4^T + x_3x_4^T + x_4x_3^T= \begin{pmatrix}
            1 & 0 & \cdots & 0 \\
            0 & 0 & \cdots & 0 \\
            \vdots & \vdots &\ddots &\vdots\\
            0 & 0 & \cdots & 0 \\
            \end{pmatrix},
            ~
            \mbox{ where~} 
            k \neq 6
            \]
            { and}
            \[
            \varphi_3 = x_1x_1^T + x_2x_2^T + x_1x_1^T + x_2x_1^T= \begin{pmatrix}
            1 & 0 & \cdots & 0 \\
            0 & 0 & \cdots & 0 \\
            \vdots & \vdots &\ddots &\vdots\\
            0 & 0 & \cdots & 0 \\
            \end{pmatrix},
             ~
            \mbox{ where~} 
            k = 6.
            \]
            { Let us} cyclically shift row of $\sigma_3(M')$ to obtain the matrix $N'^T$ {as follows:}
            \[
            N'^T =P_1\sigma_3(M') =\begin{pmatrix}
            \sigma_3(M')_3 \\
            \sigma_3(M')_4 \\
            \vdots \\
            \sigma_3(M')_k \\
            \sigma_3(M')_1 \\
            \sigma_3(M')_2
            \end{pmatrix}, 
            ~
            \mbox{ where~} 
            P_1 = \begin{pmatrix}
            0 & 0 & 1 &\cdots&0\\
            \vdots &\vdots &\vdots&\ddots&\vdots \\
            0 & 0 & 0 &\cdots&1\\
            1 & 0 & 0 &\cdots&0\\
            0 & 1 & 0 &\cdots&0
            \end{pmatrix},
            \]
            from which we have 
            \[
            M'N' = 
            \begin{pmatrix}
            0 \\
            0 & 0 \\
            0 & 0 &0 \\
            
            0&1 & \cdots & 0\\
            \vdots &\vdots&\ddots &\vdots&\ddots \\
            0 &0& \cdots & 1 &0 &0\\
            0&0&\cdots & 0& 0 &1&0&\\
            0& 0&\cdots& 0 &1&0 &0& 0 \\
            \end{pmatrix}.
            \]
            It is easy to check that \(I_k + M'N'\) is invertible.

            \item If $\varepsilon = \theta +1-|A| =0$ for $|A|=\theta+1$, then we have  
            \[
            \Delta = F_1F_1^T =
            \left(
            \begin{matrix}
            0 & 1 &  \\ 
            1 & 0 &  \\ 
             &  & 0 & \cdots &0\\
             & &\vdots & \ddots  & \vdots \\
             &  & 0 & \cdots & 0  \\
            \end{matrix}
            \right)
            \hspace{-0.4em}
            \raisebox{-1.8em}{$\left.\rule{0pt}{2.5em}\right\}$}
            \hspace{0.3em}
            \raisebox{-1.8em}{$\text{all-zero } (k-2)\text{ rows}$}
            \]
            and 
            \[
            M'M'^T =  MM^T + \Delta =
            \begin{pmatrix}
            1 & 1 & \\
            1 & 1 & \\
             & & I_{k-2}  \\
            \end{pmatrix}.
            \]
            Let $x_7 = (1,1,0,\cdots,0,1)^T$ and $x_8 = (0,0,0,\cdots,0,1)^T$. We denote by \(\sigma_4(M')\) the matrix resulting from swapping \(x_7\) and \(x_8\) in \(M'\). Let 
            \[
            \varphi_4 = x_7x_7^T + x_8x_8^T + x_7x_8^T + x_8x_7^T= \begin{pmatrix}
            1 & 1 &0 & \cdots & 0 \\
            1 & 1 & 0&\cdots & 0 \\
            0 & 0 & 0&\cdots & 0 \\
            \vdots & \vdots& \vdots &\ddots &\vdots\\
            0 & 0 &0& \cdots & 0 \\
            \end{pmatrix}.
            \]
            { Let us} cyclically shift row of $\sigma_4(M')$ to obtain the matrix $N'^T$ {as follows.}
            
            \[
            N'^T =P_2\sigma_4(M') =\begin{pmatrix}
            \sigma_4(M')_2 \\
            \sigma_4(M')_3 \\
            \vdots \\
            \sigma_4(M')_k \\
            \sigma_4(M')_1 
            \end{pmatrix}, 
            ~
            \mbox{ where~} 
            P_2 = \begin{pmatrix}
            0 & 1 & 0 &\cdots&0\\
            0 & 0 & 1 &\cdots&0\\
            \vdots &\vdots &\vdots&\ddots&\vdots \\
            0 & 0 & 0 &\cdots&1\\
            1 & 0 & 0 &\cdots&0
            \end{pmatrix}.
            \]
            Then we have 
            \[
            M'N' = 
            \begin{pmatrix}
            0 \\
            0 & 0 \\
            0&1& 0\\
            \vdots &\ddots &\ddots&\ddots \\
            0 & \cdots & 0 &1&0\\
            \end{pmatrix},
            \]
            which implies that \(I_k + M'N'\) is invertible.
        \end{enumerate}
    \end{enumerate}
\end{enumerate}

Now, let $\D$ be the code generated by $H = [N' : I_{n-k}]$. Since $I_k + M'N'$ is invertible, $(\C, \D)$ forms an LCP. Moreover, since $N'^T$ is obtained from $M$ through column permutations and row shifts, $\D^\perp$ has the same parameters as $\C$. { This shows that} both are optimal codes with minimum distance $d$. Therefore the security parameter $\min\{d(\C), d(\D^\perp)\} = d$ achieves the theoretical maximum $d_L(n,k)$. This proves the optimality of the LCP $(\C, \D)$, completing the proof.
\end{proof}

Combining Theorem \ref{sun-newthm1}, Theorem \ref{sun-lem14}, and Theorem \ref{lem9}, we have the following corollary:

\begin{corollary}  For every $k\ge 2$ and $d\ge 2^{k-1}$, then there exists  {an optimal} binary $[g_2(k,d), k]$ LCP  of Solomon–Stiffler or Belov type  with security parameters $d$. 
\end{corollary}

{
\begin{remark}
  {\rm
(1)  For $h\le   \lceil \frac{k-1}{2} \rceil$, the LCP $(\C, \D)$ obtained in Theorem 13 are all new because Güneri only consider the case of $h\ge \lceil \frac{k-1}{2} \rceil$. 

(2) Theorem~9 constructs optimal binary LCPs from Solomon--Stiffler (SS) codes, and its applicability is restricted by the condition
$
d\ge \left\lceil \frac{k-1}{2}\right\rceil 2^{k-1}.
$
In contrast, Theorem~15 constructs optimal binary LCPs from Belov codes by applying an additional puncturing operation to SS codes, namely $
d\ge 2^{k-1}.
$
Thus Theorem~15 relaxes the lower bound on \(d\) by a factor of
$
\left\lceil \frac{k-1}{2}\right\rceil.
$
Consequently, for a fixed \(k\), Theorem~15 covers all distances
\[
2^{k-1}\le d<\left\lceil \frac{k-1}{2}\right\rceil 2^{k-1},
\]
which are not covered by Theorem~9. Theorem~15 strictly extends the range of constructible optimal binary LCPs and provides more flexible parameters than Theorem~9.
  }
 \end{remark}

}

Following the proof of Theorem 15, we obtain {\bf Algorithm 1}. In Algorithm 1, the key step is to search a suitable $\sigma $ and a permutation matrix $P$ such that $I_k+MN$ is invertible.

\begin{algorithm}
\caption{Construction of Optimal Binary LCPs }
\label{alg:belov-LCPs}
\vspace{0.3\baselineskip}
\textbf{Input:} 
Target minimum distance $d = h2^{k - 1} - \sum_{i = 1}^{t} 2^{u_i - 1}(d \neq 2^{k-2})$, with $1\le u_1<\cdots<u_t<k$ and $\sum_{i=1}^{{min}\{h+1,t\}}u_{t+1-i}\le hk $


\textbf{Our customized selection for Belov puncturing:}
\begin{itemize}
\item Choose subspace $W\subseteq\mathbb{F}_2^k$ with $\dim(W)=\theta=3$ using our dedicated basis selection method.
\item Select subset $A\subset W$ such that:
    \begin{itemize}
    \item $\theta-1<|A|\le \theta+1$;
    \item every subset of at most $\theta$ vectors in $A$ is linearly independent;
    \item set $\varepsilon=\theta+1-|A|\in\{0,1\}$.
    \end{itemize}
\item Construct Belov extra puncturing set $F_1 = W\setminus A$, enforce disjoint condition $F\cap F_1=\emptyset$, where $F=\bigcup_{i=1}^t U_i$ is the column set to be deleted for Solomon-Stiffler codes.
\end{itemize}

\textbf{Output:} Optimal binary $[n,k,d]$ LCPs  $(\C,\D)$ with $n = g_2(k,d)$ and $G=[S_k(h) \setminus (F \cup F_1)]$

\textbf{ Main Steps:}
\vspace{0.3\baselineskip} 
\begin{enumerate}
\setlength{\itemsep}{2pt}
\setlength{\parskip}{2pt}
\setlength{\parsep}{2pt}
    \item \textbf{Compute  $h = \left\lceil \frac{d}{2^{k-1}} \right\rceil$, and select puncturing location for $F_1$:}
    
    \begin{enumerate}
    \setlength{\itemsep}{2pt}
        \item $h=1$: Refer to (i) in Lemma 12 and Theorem 13 for concrete choice of basis of $W$, set $F_1=W \setminus A$ with $F \cap F_1 = \emptyset$.
        \item $2 \le h \le \left\lceil \frac{k-1}{2} \right\rceil$: Refer to (ii) in Lemma 12 and Theorem 13; place deletion $F_1$ inside subspace $U_{k-4}$ of dimension $k-2$.
        \item $h \ge \left\lceil \frac{k-1}{2} \right\rceil$: Refer to Güneri’s idea and Theorem 9. 
    \end{enumerate}

    \item \textbf{$\sigma(M)$ from swap \(m_i\) and \(m_j\):}
    \begin{enumerate}
    \setlength{\itemsep}{2pt}
        \item Transform $G$ into systematic form: $G = (I_k : M)$,  compute $MM^T$. 
        \item \(\sigma(M)\) from swap \(m_i\) and \(m_j\) in \(M\)
        \item $M\sigma(M)^T = MM^T + (m_i+m_j)(m_i+m_j)^T$ 
    \end{enumerate}
    
    \item \textbf{$N^T$ from circulary shift row of $\sigma(M)$ }
    \begin{enumerate}
    \setlength{\itemsep}{2pt}
        \item circulary shift row of $\sigma(M)$ to obtain the matrix $N^T=P\sigma(M)$, where $P$ is a permutation matrix
        \item $MN = M(N^T)^T = M(P\sigma(M))^T = M\sigma(M)^TP^T$, $MN$ is a lower triangular matrix
        \item $I_k + MN$ is invertible
    \end{enumerate}
\end{enumerate}
\vspace{0.3\baselineskip}
\end{algorithm}

\begin{theorem}\label{thm17}
The time complexity of Algorithm 1 is $O(k^2 n)$, and its space complexity is $O(kn)$.   
\end{theorem}
\begin{proof}
For Algorithm 1, let the output length be $n = g_2(k,d)$. 
Let the final generator matrix be $G\in \mathbb{F}_2^{k\times n}.$
Transform $G$ into {a systematic form of 
$G=[I_k : M]$, $M\in \mathbb{F}_2^{k\times (n-k)}.$}
The matrix $N^T$ is obtained from $M$ by a fixed column swap and cyclic row shifts,
hence $N\in \mathbb{F}_2^{(n-k)\times k}.$
The algorithm finally checks whether $I_k+MN$ is invertible over $\mathbb{F}_2$ or not.

    \begin{enumerate}[label=\roman*.]
        \item[(i).] Time complexity:  
        The dominant operations are standard $\mathbb{F}_2$ linear algebra routines.
        \begin{enumerate}[label=(\arabic*)]
            \item Transforming $G$ (of size $k\times n$) into a systematic form via Gaussian elimination costs $O(k^2 n)$. 
            \item  Computing $MN$, where $M\in \mathbb{F}_2^{k\times (n-k)}$ and $N\in \mathbb{F}_2^{(n-k)\times k}$, using  matrix multiplication costs $O\bigl(k^2(n-k)\bigr)=O(k^2 n).$
            \item Invertibility testing of the $k\times k$ matrix  $I_k+MN$ by elimination costs $O(k^3)$. 
        \end{enumerate}
        The construction of $N^T$ from $M$ (column swap and cyclic shifts) is a re-indexing/copy
        operation of cost at most $O(k(n-k))=O(kn)$, which is negligible compared to the above terms.
        Therefore, the overall time complexity of Algorithm 1 is $O(k^2 n)$, or more explicitly $O(k^2 n)+O(k^3)$, with the leading term $O(k^2 n)$ when $n$ is large.
        
        \item[(ii).] Space complexity:
        The algorithm stores $G\in \Bbb F_2^{k\times n}$ and matrices $M\in \mathbb{F}_2^{k\times(n-k)}$, $N\in \mathbb{F}_2^{(n-k)\times k}$. Hence the memory usage is dominated by storing $O(kn)$ field elements. The auxiliary matrices $MN$ and $I_k+MN$ are of size $k\times k$ and require only $O(k^2)$ additional space. Consequently, the overall space complexity is $O(kn)$.
    \end{enumerate}

    This completes the proof.
\end{proof}
}


As shown in Table \ref{table:code-comparison}, Belov codes enable the construction of linear complementary pairs with smaller minimum distances $d$ compared to SS codes. They extend the parameter range and expand the family of optimal binary LCPs.

{ Table~\ref{tab:belov_lcps} lists many optimal binary LCPs obtained by our method. To extend the construction to smaller values of \(d\), we employ Belov's improvement as our core contribution: an additional puncturing operation on SS codes, namely puncturing the set \(F_1\). This enables the systematic construction of LCPs for smaller values of \(d\), thereby extending the range of achievable parameters.

As shown in Table~\ref{table:code-comparison} and Table~\ref{tab:belov_lcps}, our Belov-based approach reduces the minimum-distance threshold from 
$
\left\lceil \frac{k-1}{2}\right\rceil 2^{k-1}
$
to \(2^{k-1}\), covering a much wider range of parameters. More explicitly, for \(k=6\), Güneri's construction requires the distance threshold \(d\ge 96\) and code length \(n\ge 189\) to construct a binary optimal LCP, whereas our method requires only \(d\ge 32\) and \(n\ge 63\). This demonstrates that our method substantially improves the existing results.
}




\section{Summary and concluding remarks} 

The main contributions of this paper are the following:

\begin{enumerate}
    
    \item We demonstrated that the Belov construction, which additionally deletes columns $F_1$ beyond those removed in SS codes, yields Griesmer codes with more flexible parameters.  By carefully selecting subspaces $U_i$, the set $F_1$, and by constructing appropriate column swaps and cyclic row shifts, we proved that for every $k \geq 5$ and $d \geq 2^{k-1}$, there exists an optimal binary $[g_2(k,d), k]$ LCP  whose security parameter attains the Griesmer bound. This significantly relaxes G\"uneri's condition $d \geq \lceil (k-1)/2 \rceil 2^{k-1}$. See Theorems \ref{lem9} and \ref{sun-newthm1}.
    
    
    \item  We provided Algorithm 1, a concrete constructive procedure for generating such optimal LCPs, along with a complexity analysis showing $O(k^2 n)$ time and $O(kn)$ space. See Theorem \ref{thm17}.
    
    \item  Our construction produces many new optimal binary LCPs, as documented in Table \ref{tab:belov_lcps}, substantially expanding the known families of such codes.
\end{enumerate}


In conclusion, this work substantially expands the known families of optimal binary LCPs by integrating Belov codes into the construction framework, providing more flexible parameters and smaller minimum distances than previously possible. As a future work, it is very interesting to present more binary optimal LCPs.

\begin{table}
\begin{center} 
\caption{Comparison of Solomon-Stiffler and Belov Code Parameters for $k \geq 5$}
\label{table:code-comparison}
\begin{tabu} to 0.8\textwidth{|X[1,c]|X[2.2,c]|X[1.6,c]|} 
\hline
& Solomon-Stiffler Code \cite{G} & Belov Code \\
\hline
&$k \geq 5$,  $d \geq \lceil (k-1)/2 \rceil 2^{k-1}, h \geq 2$ & $k \geq 5, d \geq 2^{k-1}, h \geq 2$  \\
\hline
$k = 5$ & $d \geq 32, n \geq 62$ & $d \geq 16, n \geq 31$ \\
\hline
$k = 6$ & $d \geq 96, n \geq 189$ & $d \geq 32, n \geq 63$ \\
\hline
$k = 7$&  $d \geq 192, n \geq 381$ & $d \geq 64, n \geq 127$ \\
\hline
$k = 8$ & $d \geq 512, n \geq 1020$ & $d \geq 128, n \geq 255$ \\
\hline
$k = 9$ & $d \geq 1024, n \geq 2044$ & $d \geq 256, n \geq 511$ \\
\hline
$k = 10$  &$d \geq 2560, n \geq 5115$ & $d \geq 512, n \geq 1023$ \\
\hline
\end{tabu}
\end{center} 
\end{table}

\begin{table}
\centering
\caption{Binary Optimal LCPs from Belov codes  from Theorem \ref{sun-newthm1}} 
\label{tab:belov_lcps}
\begin{tabu} to 1\textwidth{|X[0.4,c]|X[2.5,c]|X[1.6,c]|} 
\hline
\textbf{$h$} & \textbf{Optimal LCPs from Belov codes} & \textbf{Minimal Security Parameter $d_{LCP}$} \\
\hline
\multirow{6}{*}{$h=1$} 
& $[20,5,9]$, $[21,5,10], \ldots ,[28,5,14]$ & $d_{\min}=9$ \\
\cline{2-3}
& $[44,6,21]$, $[45,6,22] , \ldots , [60, 6, 30]$ & $d_{\min}=21$ \\
\cline{2-3}
& $[92,7,45]$, $[93,7,46], \ldots , [124,7,62]$ & $d_{\min}=45$ \\
\cline{2-3}
& $[188,8,93]$, $[189,8,94], \ldots , [252,8,126]$ & $d_{\min}=93$ \\
\cline{2-3}
& $[380,9,189]$, $[381,9,190], \ldots , [501,9,250]$ & $d_{\min}=189$ \\
\cline{2-3}
& $[764,10,381]$, $[765,10,382], \ldots, [1013,10,506]$ & $d_{\min}=381$ \\
\hline
\multirow{6}{*}{$h=2$} 
& $[36,5,17]$, $[37,5,18], \ldots, [61,5,31]$ & when $h > 2 $, $d > 32$, stop searching\\
\cline{2-3}
& $[76,6,37]$, $[77,6,38], \ldots, [116,6,58]$ & \\
\cline{2-3}
& $[204,7,101]$, $[205,7,102], \ldots, [244,7,122]$ &  \\
\cline{2-3}
& $[412,8,205]$, $[413,8,206], \ldots, [500,8,250]$ &  \\
\cline{2-3}
& $[828,9,413]$, $[829,9,414], \ldots, [1012,9,506]$ &  \\
\cline{2-3}
& $[1660,10,829]$, $[1661,10,830], \ldots, [2036,10,1018]$ &  \\
\hline
\multirow{5}{*}{$h=3$} 
& $[163,6,81]$, $[164,6,82], \ldots, [179,6,90]$ & when $h > 3$, $d > 96$, stop searching\\
\cline{2-3}
& $[324,7,161]$, $[325,7,162], \ldots, [371,7,186]$ & when $h > 3$, $d > 192$, stop searching \\
\cline{2-3}
& $[652,8,325]$, $[653,8,326], \ldots, [755,8,378]$ &  \\
\cline{2-3}
& $[1308,9,653]$, $[1309,9,654], \ldots, [1523,9,762]$ &  \\
\cline{2-3}
& $[2620,10,1309]$, $[2621,10,1310], \ldots, [3059,10,1530]$ & \\ \cline{2-3}
\hline
\multirow{3}{*}{$h=4$}
& $[900,8,449]$, $[901,8,450], \ldots, [1010,8,506]$ & when $h > 4$, $d > 512$, stop searching \\
\cline{2-3}
& $[1804,9,901]$, $[1805,9,902], \ldots, [2034,9,1018]$ & when $h > 4$, $d > 1024$, stop searching \\
\cline{2-3}
& $[3612,10,1805]$, $[3613,10,1806], \ldots, [4082,10,2042]$ & \\
\cline{2-3}
\hline
\multirow{1}{*}{$h=5$}
& $[4620,10,2309]$, $[4621,10,2310], \ldots, [5105,10,2554]$ & when $h > 5$, $d > 2560$, stop searching \\
\cline{2-3}
\hline
\end{tabu}
\end{table}

\end{document}